\documentclass[twocolumn,superscriptaddress,pra,aps,10pt,nofootinbib]{revtex4-1}
\usepackage{times}
\usepackage{amssymb}
\usepackage{color}
\usepackage{amsmath}
\usepackage{mathrsfs}
\usepackage{amsbsy}
\usepackage{amsthm}
\usepackage{graphicx}
\usepackage{bbm}
\usepackage{bm}
\usepackage{epsfig}
\usepackage{xfrac}
\usepackage{xcolor}

\definecolor{myurlcolor}{rgb}{0,0,0.7}
\definecolor{myrefcolor}{rgb}{0.8,0,0}
\usepackage[unicode=true,pdfusetitle, bookmarks=true,bookmarksnumbered=false,bookmarksopen=false, breaklinks=false,pdfborder={0 0 0},backref=false,colorlinks=true, linkcolor=myrefcolor,citecolor=myurlcolor,urlcolor=myurlcolor]{hyperref}

\usepackage{braket}
\newcommand{\gv}[1]{\ensuremath{\text{\boldmath$ #1 $}}}
\newcommand{\abs}[1]{\left| #1 \right|} 
\newcommand{\norm}[1]{\left\| #1 \right\|} 
\newcommand{\trace}{\mathrm{Tr}}

\newcommand{\vx}{{\gv{x}}}
\newcommand{\vy}{{\gv{y}}}

\newcommand{\bE}{{\mathbb{E}}}

\renewcommand{\epsilon}{\varepsilon}
\newcommand{\appropto}{\mathrel{\vcenter{
  \offinterlineskip\halign{\hfil$##$\cr
    \propto\cr\noalign{\kern2pt}\sim\cr\noalign{\kern-2pt}}}}}

\let\baraccent=\= 
\renewcommand{\=}[1]{\stackrel{#1}{=}} 

\newcommand{\thmref}[1]{\hyperref[#1]{Theorem~\ref{#1}}}
\newcommand{\lemmaref}[1]{\hyperref[#1]{Lemma~\ref{#1}}}
\newcommand{\figref}[1]{\hyperref[#1]{Fig.~\ref{#1}}}
\newcommand{\figaref}[1]{\hyperref[#1]{Fig.~\ref{#1}a}}
\newcommand{\figbref}[1]{\hyperref[#1]{Fig.~\ref{#1}b}}
\newcommand{\figcref}[1]{\hyperref[#1]{Fig.~\ref{#1}c}}
\renewcommand{\eqref}[1]{\hyperref[#1]{Eq.~(\ref{#1})}}
\newcommand{\eqsref}[2]{\hyperref[#1]{Eqs.~(\ref{#1})-(\ref{#2})}}
\newcommand{\appref}[1]{\hyperref[#1]{Appx.~\ref{#1}} of \cite{SM}}

\usepackage{dsfont}

\newtheorem{theorem}{Theorem}
\newtheorem{lemma}{Lemma}

\begin{document}

\title{
An exact correspondence between the quantum Fisher information and the Bures metric
}

\author{Sisi Zhou}
\affiliation{Departments of Applied Physics and Physics, Yale University, New Haven, Connecticut 06511, USA}
\affiliation{Yale Quantum Institute, Yale University, New Haven, Connecticut 06511, USA}
\affiliation{Pritzker School of Molecular Engineering, The University of Chicago, Illinois 60637, USA}

\author{Liang Jiang}
\affiliation{Departments of Applied Physics and Physics, Yale University, New Haven, Connecticut 06511, USA}
\affiliation{Yale Quantum Institute, Yale University, New Haven, Connecticut 06511, USA}
\affiliation{Pritzker School of Molecular Engineering, The University of Chicago, Illinois 60637, USA}

\date{\today}

\begin{abstract}
The quantum information and the Bures metric are equivalent to each other, except at points where the rank of the density matrix changes. Here we show that by slightly modifying the definition of the Bures metric, the quantum information will be fully equivalent to the Bures metric without exception. 
\end{abstract}

\maketitle

\onecolumngrid

\section{Introduction} 

Quantum Fisher information (QFI) is an important concept in quantum metrology, working as a good measure of the estimation precision of an unknown parameter $x$ in an density matrix $\rho_x$. QFI appears in the famous quantum Cram\'{e}r-Rao bound~\cite{helstrom1976quantum,holevo2011probabilistic,braunstein1994statistical}, 
\begin{equation}
\delta^2 x \geq \frac{1}{N_{\rm expr} \cdot F(\rho_x)},
\end{equation}
where $\delta^2 x = \bE[(\hat{x} - x)^2]$ is the variance of the estimator $\hat{x}$ of an unknown parameter $x$, $N_{\rm expr}$ is the number of experiments (the number of $\rho_x$ used) and $F(\rho_x)$ is the QFI. For multi-parameter estimation, consider $\vx=(x_1\cdots x_P)^T$ where $P$ is the number of parameters. , we also have 
\begin{equation}
\delta^2 \vx \succeq \frac{1}{N_{\rm expr}} \cdot F(\rho_\vx)^{-1},
\end{equation}
where $\delta^2 \vx = \bE[(\hat{\vx} - \vx)(\hat{\vx} - \vx)^T]$ is the covariance matrix and $F(\rho_\vx)$ is the QFI matrix. ``$A \succeq B$'' here means $A - B$ is a positive semidefinite matrix. The quantum Cram\'{e}r-Rao bound is asyptotically saturable ($N \rightarrow \infty$) in the single-parameter case and not always saturable in the multi-parameter case, while the QFI matrix still provides an efficient lower bound of the estimation error. 

In this paper, we consider only $\rho_\vx$ living in finite dimentional Hilbert spaces. Using the diagonal form of the density matrix $\rho_\vx = \sum_{i=1}^d \lambda_i \ket{k}\bra{k}$ where $d$ is the dimension of the Hilbert space, the QFI matrix is defined by 
\begin{equation}
\label{eq:QFI}
F^{ij}(\rho_\vx) = 2 \sum_{k,\ell:\lambda_k + \lambda_\ell > 0} \frac{\mathrm{Re}[\bra{k}\partial_i \rho_\vx \ket{\ell}\bra{\ell} \partial_j \rho_\vx \ket{k}]}{\lambda_k + \lambda_\ell},
\end{equation}
where $i,j = 1,\ldots, P$.

The connection between the QFI and the Bures metric~\cite{hubner1992explicit,sommers2003bures,vsafranek2017discontinuities} was well recognized and widely applied (see e.g.~\cite{escher2011general,zhong2013fisher,yuan2016sequential}). It was believed that 
\begin{equation}
\frac{1}{4}\sum_{i,j=1}^P F^{ij}(\rho_\vx)dx_i dx_j \stackrel{\text{?}}{=} \sum_{i,j=1}^P g^{ij}(\rho_\vx) dx_i dx_j := d_B^2 (\rho_\vx,\rho_{\vx+d\vx})
\end{equation}
where $d_{B}^2 (\rho_1,\rho_2)$ is the Bures distance between $\rho_1$ and $\rho_2$, defined by $d_{B}^2 (\rho_1,\rho_2) = \sqrt{2(1 - F_B(\rho_1,\rho_2))}$ and the fidelity $F_B(\rho_1,\rho_2) = \trace\big(\sqrt{\sqrt{\rho_1}\rho_2\sqrt{\rho_1}}\big)$. However, it was shown that~\cite{vsafranek2017discontinuities} for any second order smooth (the first and second derivatives both exist and are continuous) function $\rho_\vx$,
\begin{equation}
4g^{ij}(\rho_\vx) = F^{ij}(\rho_\vx) + 2\sum_{k:\lambda_k = 0} \partial_{i}\partial_{j}\lambda_k,
\end{equation}
and $F(\rho_\vx) = 4g(\rho_\vx)$ if and only if for all $k$ and $\vx$ such that $\lambda_k = 0$, the Hessian matrices of $\lambda_k$ is zero. Consider $\rho_x = x^2\ket{0}\bra{0} + (1-x^2) \ket{1}\bra{1}$, we can calculate the Bures metric and the QFI at $x = 0$, which gives $F(\rho_x) = 0$ and $g(\rho_x) = 1$. It provides a simple example showing the discrepancy between the QFI and the Bures metric. 

To circumvent this discrepancy, we propose a modified definition of the Bures metric $h^{ij}(\rho_\vx)$,
\begin{equation}
\label{eq:m-bures}
\sum_{i,j=1}^P h^{ij}(\rho_\vx) dx_i dx_j := d_B^2 (\rho_{\vx-\frac{1}{2}d\vx},\rho_{\vx+\frac{1}{2}d\vx})
\end{equation}
and will show that $F(\rho_\vx) = 4h(\rho_\vx)$ for any second order differentiable $\rho_\vx$. Consider again the example where $\rho_x = x^2\ket{0}\bra{0} + (1-x^2) \ket{1}\bra{1}$, we can easily verify that $h(\rho_x) = 0$ and $F(\rho_x) = 4h(\rho_x)$. We will prove that such a correspondence between $F(\rho_\vx)$ and $h(\rho_\vx)$ is general. 

{
Note that here our discussion is based on the assumption that $\rho_\vx$ is well defined in the neighborhood of $\vx$. However, one should be careful with situations where $\vx$ is in the boundary of a closed set in $\mathbb{R}^N$, for example, $\rho_x = x^2\ket{0}\bra{0} + (1-x^2) \ket{1}\bra{1}$ defined on $x \in[0,1]$. In this case, \eqref{eq:QFI} is ill-defined at $x=0$ and $F(\rho_0) = 0$ does not capture the estimation precision of $x$. 
}

\section{Results}

In this section, we will provide a proof of the following theorem:
\begin{theorem}
\label{thm:main}
$F(\rho_\vx) = 4h(\rho_\vx)$ for any second order differentiable $\rho_\vx$. 
\end{theorem}
Before going into the details of the proof, we first state three useful lemmas and provide a proof of Lemma 3.
\begin{lemma}[Distance between two matrices~\cite{mirsky1960symmetric}]
\label{lemma:dist}
Let $\alpha_1 \geq \alpha_2 \geq \cdots \geq \alpha_n$ and $\beta_1 \geq \beta_2 \geq \cdots \geq \beta_n$ be the singular values of the complex matrices $M_1$ and $M_2$ respectively. Then 
\begin{equation}
\norm{M_1 - M_2} \geq \norm{{\rm diag}(\alpha_1-\beta_1,\ldots,\alpha_n-\beta_n)},
\end{equation}
for any unitarily invariant norm $\norm{\cdot}$. 
\end{lemma}
\begin{lemma}[Taylor expansion of the matrix square root function \cite{del2018taylor}]
\label{lemma:taylor}
If $A \succ 0$ and $A + H\succ 0$ (``$M_1 \succ M_2$'' here means $M_1 - M_2$ is a positive matrix), then we have the Taylor expansion of $\sqrt{A+H}$ up to the $n$-th order,
\begin{equation}
\sqrt{A + H} = \sqrt{A} + \sum_{1\leq k \leq n} \frac{1}{k!}\nabla^k(\sqrt{A}) \cdot H + O(\norm{H}_2^{n+1})
\end{equation}
where $\norm{\cdot}_2$ is the $L_2$ norm, $\nabla^n(\sqrt{A})$ is the $n$-th order derivative operator at $A$ defined by 
\begin{equation}
\nabla(\sqrt{A}) \cdot H = \int_0^\infty e^{-s\sqrt{A}} H e^{-s\sqrt{A}} ds,
\end{equation}
when $n=1$ and for $n\geq 2$,
\begin{equation}
\nabla^n(\sqrt{A}) \cdot H = - \nabla(\sqrt{A}) \cdot \bigg[ \sum_{\substack{p+q=n-2,\\p,q\geq 0}} \frac{n!}{(p+1)!(q+1)!} \big[\nabla^{p+1}(\sqrt{A}) \cdot H\big] \big[\nabla^{q+1}(\sqrt{A}) \cdot H\big] \bigg].
\end{equation}
\end{lemma}

\begin{lemma}
\label{lemma:singular}
Consider a positive semidefinite matrix 
\begin{equation}
M(\delta) = 
\begin{pmatrix}
A & \delta B\\
\delta B^\dagger  & \delta^2 C
\end{pmatrix} + \begin{pmatrix}
o(\delta) & o(\delta)\\
o(\delta)  & o(\delta^2)
\end{pmatrix},
\end{equation}
as a function defined over $\delta \in [0,a)$ for some $a > 0$. $A,B,C = O(1)$ are matrices satisfying
\begin{equation}
A \succ 0, \quad 
C - B^\dagger A^{-1} B \succeq 0.
\end{equation}
Then 
\begin{equation}
\trace\left(\sqrt{M(\delta)}\right) = \trace\big(\sqrt{A}\big) + \delta \trace\big(\sqrt{C - B^\dagger A^{-1} B}\big) + o(\delta).
\end{equation}
\end{lemma}

\begin{proof}
Let $G$ be a Hermitian matrix
\begin{equation}
G = \begin{pmatrix}
0 & G_{12} \\
G_{21} & 0 \\
\end{pmatrix}.
\end{equation}
Then 
\begin{equation}
e^{i\delta G}
\begin{pmatrix}
A & \delta B\\
\delta B^\dagger  & \delta^2 C
\end{pmatrix}
 e^{-i\delta G}
= \begin{pmatrix}
A  & 0 \\
0  & \delta^2 (C - B^\dagger A^{-1} B)
\end{pmatrix}
+ \begin{pmatrix}
o(\delta) & o(\delta)\\
o(\delta)  & o(\delta^2)
\end{pmatrix},
\end{equation}
\begin{equation}
\Leftrightarrow~ 
\begin{pmatrix}
A & \delta B\\
\delta B^\dagger  & \delta^2 C
\end{pmatrix}
= \left(I - i\delta G - \frac{\delta^2}{2} G^2\right) \begin{pmatrix}
A & 0 \\
0 & \delta^2 (C - B^\dagger A^{-1} B) 
\end{pmatrix} \left(I + i\delta G - \frac{\delta^2}{2} G^2\right)
+ \begin{pmatrix}
o(\delta) & o(\delta)\\
o(\delta)  & o(\delta^2)
\end{pmatrix},
\end{equation}
\begin{equation}
\Leftrightarrow~ 
\begin{pmatrix}
0 & B \\
B^\dagger & 0
\end{pmatrix}
= 
-i\delta \begin{pmatrix}
0 & - A G_{12} \\
G_{21} A & 0 
\end{pmatrix},\quad
C = 
C - B^\dagger A^{-1} B + G_{21} A G_{12},
\end{equation}
\begin{equation}
\Leftrightarrow~
 G_{12} = -i A^{-1}B,\quad
 G_{21} = i B^\dagger A^{-1}.
\end{equation}
Using the same technique, it is easy to show that there exists a Hermitian matrix $G' = o(\delta)$ such that


\begin{equation}
e^{i G'} e^{i\delta G}
\begin{pmatrix}
A & \delta B\\
\delta B^\dagger  & \delta^2 C
\end{pmatrix}
 e^{-i\delta G} e^{-i G'}
= \begin{pmatrix}
A  & 0 \\
0  & \delta^2 (C - B^\dagger A^{-1} B)
\end{pmatrix}
+ \begin{pmatrix}
o(\delta) & o(\delta^2) \\
o(\delta^2) & o(\delta^2)
\end{pmatrix},
\end{equation}
Therefore we have
\begin{equation}
\begin{split}
\trace\left(\sqrt{\begin{pmatrix}
A & \delta B\\
\delta B^\dagger  & \delta^2 C
\end{pmatrix}}\right) 
&= \trace\left(\sqrt{\begin{pmatrix}
A + o(\delta)  & 0 \\
0  & \delta^2 (C - B^\dagger A^{-1} B) 
\end{pmatrix}
+ \begin{pmatrix}
0 & o(\delta^2) \\
o(\delta^2) & o(\delta^2)
\end{pmatrix}
}\right) \\
&= \trace\left(\sqrt{\begin{pmatrix}
A  & 0 \\
0  & \delta^2 (C - B^\dagger A^{-1} B)
\end{pmatrix}
}\right) + o(\delta),
\end{split}
\end{equation}
where in the last step, we used \lemmaref{lemma:dist}. 
\end{proof}

To prove \thmref{thm:main}, we consider the derivative along a direction $\vy$. For any second order differentiable function $\rho_\vx$, we have the Taylor expansion of $\rho_{\vx+\epsilon \vy}$ equal to 
\begin{equation}
\rho_{\vx + \epsilon \vy} 
= \rho_{\vx} + \epsilon \sum_{i=1}^P (\partial_i \rho_{\vx}) y_i + \frac{\epsilon^2}{2}\sum_{i,j=1}^P (\partial_i\partial_j \rho_{\vx}) y_i y_j + o(\epsilon^2),
\end{equation}
in a neighbourhood of $\epsilon = 0$. Using the simplified notation $\rho(\epsilon):=\rho_{\vx + \epsilon \vy}$, $\rho_\vx = \Lambda$, $R := \sum_{i=1}^P (\partial_i \rho_{\vx}) y_i$ and $S:= \frac{1}{2}\sum_{i,j=1}^P (\partial_i\partial_j \rho_{\vx}) y_i y_j$. Our goal is to prove
\begin{equation}
\label{eq:main}
\trace(\sqrt{\sqrt{\rho(\epsilon)}\rho(-\epsilon)\sqrt{\rho(\epsilon)}}) = \trace(\Lambda) + \epsilon^2 \trace(S) - \epsilon^2 \sum_{k,\ell:\lambda_k+\lambda_\ell \neq 0} \frac{\abs{R_{k\ell}}^2}{\lambda_k+\lambda_\ell} + o(\epsilon^2).
\end{equation}
If \eqref{eq:main} holds, we will have 
\begin{equation}
\begin{split}
\trace(\sqrt{\sqrt{\rho(\epsilon)}\rho(-\epsilon)\sqrt{\rho(\epsilon)}}) 
&= 1 - \epsilon^2 \sum_{k,\ell:\lambda_k+\lambda_\ell \neq 0} \frac{\big| \sum_{i=1}^P \bra{k}\partial_i \rho_{\vx}\ket{\ell} y_i \big|^2}{\lambda_k+\lambda_\ell}\\
&= 1 - \epsilon^2 \sum_{i,j=1}^P y_i y_j \sum_{k,\ell:\lambda_k+\lambda_\ell \neq 0}  \frac{\bra{k}\partial_i \rho_{\vx}\ket{\ell}\bra{\ell}\partial_j \rho_{\vx}\ket{k}}{\lambda_k + \lambda_\ell}\\
&= 1 - \frac{\epsilon^2}{2} \sum_{i,j=1}^P y_i y_j F^{ij}(\rho_\vx),
\end{split}
\end{equation}
which implies \thmref{thm:main}. Therefore, in order to prove \thmref{thm:main}, it will be sufficient to prove the following slightly more general theorem:
\begin{theorem}
Let $\rho(\epsilon)$ be a positive semidefinite matrix equal to 
\begin{equation}
\rho(\epsilon) = \Lambda + \epsilon R + \epsilon^2 S + o(\epsilon^2)
\end{equation}
in a neighbourhood of $\epsilon  = 0$, where $\Lambda$ is a diagonal matrix with diagonal elements $\{\lambda_k\}_{k=1}^d$. Then \eqref{eq:main} holds true. 
\end{theorem}
\begin{proof}
Choose a proper order of basis such that $\Lambda = \begin{pmatrix}
\Lambda_+ & 0 \\
0 & 0 
\end{pmatrix}$ where $\Lambda_+$ is a positive diagonal matrix. Let $R = \begin{pmatrix}
R_{11} & R_{12} \\
R_{21} & R_{22} 
\end{pmatrix}$ and  $S = \begin{pmatrix}
S_{11} & S_{12} \\
S_{21} & S_{22} 
\end{pmatrix}$. Due to the positivity of $\rho(\epsilon)$ around $\epsilon = 0$, we must have $R_{22} = 0$ and $S_{22} - R_{21} \Lambda_+^{-1} R_{12} \succeq 0$. 

Now we choose a Hermitian matrix $G = \begin{pmatrix}
0 & G_{12} \\
G_{21} & 0 \\
\end{pmatrix}$ such that 
\begin{equation}
\begin{split}
\rho(\epsilon) 
&= \Lambda + \epsilon R + \epsilon^2 S + o(\epsilon^2) 
= e^{-i\epsilon G} \Lambda(\epsilon) e^{i\epsilon G},
\end{split}
\end{equation}
where 
\begin{equation}
\Lambda(\epsilon) = \begin{pmatrix}
O(1) & O(\epsilon^2) \\
O(\epsilon^2) & O(\epsilon^2)
\end{pmatrix}.
\end{equation}
A simple calculation shows that $G_{12} = - i \Lambda_+^{-1}R_{12}, G_{21} = i R_{21} \Lambda_+^{-1}$ is a proper choice, leading to  
\begin{equation}
\Lambda(\epsilon) = 
\begin{pmatrix}
\Lambda_{+} + \epsilon R_{11} + \epsilon^2 T_{11} & \epsilon^2 T_{12} \\
\epsilon^2 T_{21} &  \epsilon^2 T_{22} 
\end{pmatrix}
+ \begin{pmatrix}
o(\epsilon^2) & o(\epsilon^2)\\
o(\epsilon^2) & o(\epsilon^2)
\end{pmatrix},
\end{equation}
where
$T_{11} = S_{11} + \frac{1}{2}(G_{12}G_{21}\Lambda_+ + \Lambda_+ G_{12}G_{21})$ and $T_{22} = S_{22} - R_{21} \Lambda_+^{-1} R_{12}$. The values of $T_{12}$ and $T_{21}$ will not affect the results. 
It is easy to verify that 
\begin{equation}
\sqrt{\Lambda(\epsilon)} = 
\begin{pmatrix}
\sqrt{\Lambda_{+} + \epsilon R_{11} + \epsilon^2 T_{11}} & \epsilon^2 \Lambda_+^{-1/2} T_{12} \\
\epsilon^2 T_{21} \Lambda_+^{-1/2} &  \abs{\epsilon} \sqrt{T_{22}} 
\end{pmatrix}
+ \begin{pmatrix}
o(\epsilon^2) & o(\epsilon^2)\\
o(\epsilon^2) & o(\epsilon)
\end{pmatrix},
\end{equation}
then 
\begin{equation}
\begin{split}
&\quad~~
e^{i \epsilon G}\sqrt{\rho(\epsilon)} \rho(-\epsilon) \sqrt{\rho(\epsilon)} e^{-i\epsilon G}\\
&= 
\sqrt{\Lambda(\epsilon)} 
\Big(I + 2 i\epsilon G - 2 \epsilon^2 G^2\Big) 
\Lambda(-\epsilon) 
\Big(I - 2 i\epsilon G - 2 \epsilon^2 G^2\Big) 
\sqrt{\Lambda(\epsilon)}
 + \begin{pmatrix}
o(\epsilon^2) & o(\epsilon^2) \\
o(\epsilon^2) & o(\epsilon^4)
\end{pmatrix}  \\
&= \sqrt{\Lambda(\epsilon)} \Lambda(-\epsilon) \sqrt{\Lambda(\epsilon)} +2i\epsilon  \sqrt{\Lambda(\epsilon)} (G \Lambda(-\epsilon) - \Lambda(-\epsilon) G ) \sqrt{\Lambda(\epsilon)}  \\
&\qquad\; - 2\epsilon^2 \sqrt{\Lambda(\epsilon)} ( G^2 \Lambda(-\epsilon) + \Lambda(-\epsilon) G^2 - 2G \Lambda(-\epsilon) G)  \sqrt{\Lambda(\epsilon)}  + \begin{pmatrix}
o(\epsilon^2) & o(\epsilon^2) \\
o(\epsilon^2) & o(\epsilon^4)
\end{pmatrix}.
\end{split}
\end{equation}
Let $\epsilon \geq 0$, we calculate each term in detail 
(the $\begin{pmatrix}
o(\epsilon^2) & o(\epsilon^2) \\
o(\epsilon^2) & o(\epsilon^4)
\end{pmatrix}$ term is omitted in Term (1) (2) and (3) for simplicity): 

\noindent Term (1):
\begin{equation}
\begin{split}
\sqrt{\Lambda(\epsilon)} \Lambda(-\epsilon) \sqrt{\Lambda(\epsilon)}
&\approx \begin{pmatrix}
\sqrt{\Lambda_{+} + \epsilon R_{11} + \epsilon^2 T_{11}}(\Lambda_{+} - \epsilon R_{11} + \epsilon^2 T_{11})\sqrt{\Lambda_{+} + \epsilon R_{11} + \epsilon^2 T_{11}}& \epsilon^2 \Lambda_+ T_{12}\\
\epsilon^2 T_{21} \Lambda_+ & \epsilon^4 (T_{22}^2 + T_{21}T_{12})\\
\end{pmatrix},
\end{split}
\end{equation}
Term (2):
\begin{equation}
2i\epsilon  \sqrt{\Lambda(\epsilon)} (G \Lambda(-\epsilon) - \Lambda(-\epsilon) G ) \sqrt{\Lambda(\epsilon)} \approx 2i\epsilon^2
\begin{pmatrix}
0 & - \Lambda_+^{3/2} G_{12} \sqrt{T_{22}}  \\
\sqrt{T_{22}} G_{21} \Lambda_+^{3/2} & ( \sqrt{T_{22}} G_{21} \Lambda_+^{1/2} T_{12} -  T_{21} \Lambda_+^{1/2} G_{12}\sqrt{T_{22}} ) \epsilon^2 \\
\end{pmatrix},
\end{equation}
Term (3):
\begin{multline}
- 2\epsilon^2 \sqrt{\Lambda(\epsilon)} ( G^2 \Lambda(-\epsilon) + \Lambda(-\epsilon) G^2 - 2G \Lambda(-\epsilon) G)  \sqrt{\Lambda(\epsilon)}  \approx \\ -2\epsilon^2 
\begin{pmatrix}
\sqrt{\Lambda_+} G_{12}G_{21} \Lambda_+^{3/2} + \Lambda_+^{3/2} G_{12}G_{21} \sqrt{\Lambda_+} & 0 \\
0 & - 2 \sqrt{T_{22}} G_{21} \Lambda_+ G_{12} \sqrt{T_{22}} \epsilon^2 \\
\end{pmatrix}.
\end{multline}
Then using \lemmaref{lemma:singular} (taking $\delta = \epsilon^2$) and \lemmaref{lemma:taylor}, we have 
\begin{equation}
\begin{split}
&\quad~\, \trace\Big(\sqrt{\sqrt{\rho(\epsilon)} \rho(-\epsilon) \sqrt{\rho(\epsilon)}}\Big) = \trace\Big( e^{-i\epsilon G}\sqrt{\sqrt{\rho(\epsilon)} \rho(-\epsilon) \sqrt{\rho(\epsilon)}} e^{i\epsilon G}\Big)\\
&= \trace\Big(\big( \sqrt{\Lambda_{+} + \epsilon R_{11} + \epsilon^2 T_{11}}(\Lambda_{+} - \epsilon R_{11} + \epsilon^2 T_{11})\sqrt{\Lambda_{+} + \epsilon R_{11} + \epsilon^2 T_{11}} \\
&\qquad \qquad \qquad \qquad \qquad - 2\epsilon^2(\sqrt{\Lambda_+} G_{12}G_{21} \Lambda_+^{3/2} + \Lambda_+^{3/2} G_{12}G_{21} \sqrt{\Lambda_+})\big)^{1/2}\Big)
+ \epsilon^2 \trace(T_{22}) + o(\epsilon^2)\\
&= \trace(\sqrt{\Lambda_+}) + \epsilon^2 \trace(T_{11}) - \epsilon^2 \sum_{i,j:\lambda_i,\lambda_j > 0} \frac{\abs{R_{ij}}^2}{\lambda_i+\lambda_j} - 2 \epsilon^2 \trace(G_{12}G_{21}\Lambda_+) + \epsilon^2 \trace(S_{22} - R_{21} \Lambda_+^{-1} R_{12}) + o(\epsilon^2)\\
&= \trace(\sqrt{\Lambda_+}) + \epsilon^2 \trace(S) - \epsilon^2 \sum_{i,j:\lambda_i > 0,\lambda_j > 0} \frac{\abs{R_{ij}}^2}{\lambda_i+\lambda_j} - 2\epsilon^2 \trace(R_{21} \Lambda_+^{-1} R_{12}) + o(\epsilon^2)\\
&= \trace(\sqrt{\Lambda}) + \epsilon^2 \trace(S) - \epsilon^2 \sum_{i,j:\lambda_i + \lambda_j > 0} \frac{\abs{R_{ij}}^2}{\lambda_i+\lambda_j} + o(\epsilon^2).
\end{split}
\end{equation}

\end{proof}

\section{Conclusions}

We put forward a new definition of the Bures metric which is fully compatible with the QFI, as opposed to the previous one where discrepancy exists in some singular points. It also provides a more reliable approach to calculate the QFI numerically using the Bures metric. 

\section{Acknowledgements}

We thank Kyungjoo Noh, Rafa{\l} Demkowicz-Dobrza\'{n}ski, Zhou Fan, Jing Yang, Yuxiang Yang for helpful discussions. We acknowledge support from the ARL-CDQI (W911NF15-2-0067, W911NF-18-2-0237), ARO (W911NF-18-1-0020, W911NF-18-1-0212), ARO MURI (W911NF-16-
1-0349), AFOSR MURI (FA9550-15-1-0015), DOE (DE-SC0019406), NSF (EFMA-1640959), and the Packard Foundation (2013-39273).

\bibliographystyle{aps}

\begin{thebibliography}{10}
\providecommand{\url}[1]{\texttt{#1}}
\providecommand{\urlprefix}{URL }
\providecommand{\eprint}[2][]{\url{#2}}

\bibitem{helstrom1976quantum}
C.~W. Helstrom, \emph{Quantum detection and estimation theory}  (Academic press
  1976).

\bibitem{holevo2011probabilistic}
A.~S. Holevo, \emph{Probabilistic and statistical aspects of quantum theory},
  volume~1  (Springer Science \& Business Media 2011).

\bibitem{braunstein1994statistical}
S.~L. Braunstein and C.~M. Caves, Statistical distance and the geometry of
  quantum states, Physical Review Letters \textbf{72}, 3439 (1994).

\bibitem{hubner1992explicit}
M.~H{\"u}bner, Explicit computation of the bures distance for density matrices,
  Physics Letters A \textbf{163}, 239 (1992).

\bibitem{sommers2003bures}
H.-J. Sommers and K.~Zyczkowski, Bures volume of the set of mixed quantum
  states, Journal of Physics A: Mathematical and General \textbf{36}, 10083
  (2003).

\bibitem{vsafranek2017discontinuities}
D.~{\v{S}}afr{\'a}nek, Discontinuities of the quantum fisher information and
  the bures metric, Physical Review A \textbf{95}, 052320 (2017).

\bibitem{escher2011general}
B.~Escher, R.~de~Matos~Filho, and L.~Davidovich, General framework for
  estimating the ultimate precision limit in noisy quantum-enhanced metrology,
  Nature Physics \textbf{7}, 406 (2011).

\bibitem{zhong2013fisher}
W.~Zhong, Z.~Sun, J.~Ma, X.~Wang, and F.~Nori, Fisher information under
  decoherence in bloch representation, Physical Review A \textbf{87}, 022337
  (2013).

\bibitem{yuan2016sequential}
H.~Yuan, Sequential feedback scheme outperforms the parallel scheme for
  hamiltonian parameter estimation, Physical review letters \textbf{117},
  160801 (2016).

\bibitem{mirsky1960symmetric}
L.~Mirsky, Symmetric gauge functions and unitarily invariant norms, The
  quarterly journal of mathematics \textbf{11}, 50 (1960).

\bibitem{del2018taylor}
P.~Del~Moral and A.~Niclas, A taylor expansion of the square root matrix
  function, Journal of Mathematical Analysis and Applications \textbf{465}, 259
  (2018).

\end{thebibliography}

\end{document}